\documentclass[pmlr]{jmlr}

\usepackage{mathtools}
\usepackage{comment}

\usepackage{./sty/tikzit}
\usetikzlibrary{shapes}
\usetikzlibrary{arrows}
\usetikzlibrary{positioning}
\usetikzlibrary{automata}
\usetikzlibrary{shapes.geometric}
\usetikzlibrary{commutative-diagrams}

\tikzstyle{node}=[fill=black, draw=black, shape=circle, scale=0.5]
\tikzstyle{wnode}=[fill=white, draw=black, shape=circle, scale=0.5]
\tikzstyle{textbox}=[inner sep=2pt, shape=rectangle, fill=none]
\tikzstyle{textnode}=[inner sep=0mm, shape=circle, fill=white]
\tikzstyle{gnode}=[inner sep=0mm, minimum size=1mm, fill={rgb,255: red,221; green,221; blue,221}, draw={rgb,255: red,221; green,221; blue,221}, shape=circle]
\tikzstyle{refine}=[fill=black, draw=black, shape=regular polygon, regular polygon sides=3, rotate=180, scale=0.5]
\tikzstyle{coarsen}=[fill=white, draw=black, shape=regular polygon, regular polygon sides=3, scale=0.5]
\tikzstyle{bdytextbox}=[outer sep=2pt, fill=white, draw=black, shape=rectangle]
\tikzstyle{redbox}=[fill=white, draw=red, shape=rectangle, text=red]
\tikzstyle{bluecirc}=[inner sep=1mm, fill=white, draw={rgb,255: red,4; green,51; blue,255}, shape=circle, text={rgb,255: red,4; green,51; blue,255}]
\tikzstyle{rednode}=[fill=red, draw=red, shape=circle, scale=0.5]
\tikzstyle{new style 0}=[fill=white, draw=red, shape=circle, scale=0.5]
\tikzstyle{bluenode}=[fill={rgb,120: red,4; green,51; blue,120}, draw={rgb,120: red,4; green,51; blue,120}, shape=circle, scale=0.5]
\tikzstyle{yellownode}=[fill={rgb,255: red,255; green,210; blue,75}, draw={rgb,255: red,255; green,210; blue,75}, shape=circle, scale=0.5]
\tikzstyle{blacksq}=[fill=black, draw=black, shape=rectangle, scale=0.5]
\tikzstyle{bluetext}=[fill=none, draw=none, shape=rectangle, text={rgb,255: red,4; green,51; blue,255}]
\tikzstyle{reg}=[draw, fill=white, rounded rectangle, rounded rectangle left arc=none, minimum height=1em, minimum width=1em, node font={\scriptsize}]
\tikzstyle{coreg}=[draw, fill=white, rounded rectangle, rounded rectangle right arc=none, minimum height=1em, minimum width=1em, node font={\scriptsize}]
\tikzstyle{turquoisenode}=[fill={rgb,255: red,115; green,255; blue,239}, draw=black, shape=circle, scale=0.5]
\tikzstyle{resistor}=[R]
\tikzstyle{inductor}=[L]
\tikzstyle{capacitor}=[C]
\tikzstyle{voltage-source}=[american voltage source]
\tikzstyle{current-source}=[american current source]

\tikzstyle{edge}=[-, draw=black]
\tikzstyle{diredge}=[->, draw=black]
\tikzstyle{dashed edge}=[-, dashed, dash pattern=on 1pt off 1.5pt, draw=black]
\tikzstyle{dirdash}=[->, dashed, dash pattern=on 2pt off 0.5pt, draw=black]
\tikzstyle{mapsto}=[{|->}, draw=black]
\tikzstyle{gray diredge}=[draw={rgb,255: red,221; green,221; blue,221}, ->]
\tikzstyle{dark grey dirdash}=[->, dashed, dash pattern=on 2pt off 0.5pt, draw={rgb,255: red,81; green,81; blue,81}]
\tikzstyle{doubedge}=[-, draw=black, double=none, double distance=3pt, inner sep=0pt, thick]
\tikzstyle{thedge}=[-, line width=1.5pt, draw=black]
\tikzstyle{gray dashed}=[-, dashed, dash pattern=on 1pt off 1.5pt, draw={rgb,255: red,128; green,128; blue,128}]
\tikzstyle{rededge}=[-, draw=red]
\tikzstyle{gray edge}=[-, draw={rgb,255: red,128; green,128; blue,128}]
\tikzstyle{blthedge}=[-, thick, draw={rgb,255: red,4; green,51; blue,255}]
\tikzstyle{blthdash}=[-, dashed, dash pattern=on 1pt off 1.5pt, thick, draw={rgb,255: red,4; green,51; blue,255}]
\tikzstyle{dirrededge}=[draw=red, ->]

\graphicspath{ {figures/} }

\newcommand{\cs}{\mathtt{cs}}
\newcommand{\row}{\mathtt{row}}

\theorembodyfont{\upshape}
\theoremheaderfont{\scshape}
\theorempostheader{:}
\theoremsep{\newline}

\jmlrworkshop{Learning and Automata (LearnAut 2024)}

\title[Learning Closed Signal Flow Graphs]{Learning Closed Signal Flow Graphs}

  \author{\Name{Ekaterina Piotrovskaya} \Email{kate.piotrovskaya.21@ucl.ac.uk}\\
  \Name{Leo Lobski} \Email{leo.lobski.21@ucl.ac.uk}\\
   \Name{Fabio Zanasi} \Email{f.zanasi@ucl.ac.uk}\\
  \addr University College London, UK}

\begin{document}
\LinesNumbered
\maketitle

\begin{abstract}
We develop a learning algorithm for {\em closed signal flow graphs} -- a graphical model of signal transducers. The algorithm relies on the correspondence between closed signal flow graphs and {\em weighted finite automata} on a singleton alphabet. We demonstrate that this procedure results in a genuine reduction of complexity: our algorithm fares better than existing learning algorithms for weighted automata restricted to the case of a singleton alphabet.
\end{abstract}

\section{Introduction}

{\em Signal flow graphs} (SFG) are a graphical language for signal transducers, which play a foundational role in control theory and engineering~\citep{Shannon-42}. A signal flow graph is typically specified as a circuit, with gates corresponding to basic signal operations: addition, copying, amplification by a scalar, and delay. Moreover, signal flow graphs allow for the formation of feedback loops. From an expressiveness viewpoint, these models capture precisely rational functions: the Taylor expansion of the function may be regarded as the signal processed by the signal flow graph.

\begin{figure}[h]
    \centering
    \begin{tikzpicture}
	\begin{pgfonlayer}{nodelayer}
		\node [style=bdytextbox] (2) at (3, 0) {$c$};
		\node [style=none] (3) at (5, 0) {};
		\node [style=bdytextbox] (6) at (-3, 0) {$+$};
		\node [style=none] (7) at (-5, 0) {};
		\node [style=bdytextbox] (9) at (0, 1.5) {$x_0$};
		\node [style=bdytextbox] (10) at (-1.25, -1.75) {$x_0$};
		\node [style=bdytextbox] (11) at (1.25, -1.75) {$x_0$};
		\node [style=none] (12) at (-7.25, 0) {$1,0,0,\dots$};
		\node [style=none] (13) at (8, 0) {$1,1,2,3,5,\dots$};
	\end{pgfonlayer}
	\begin{pgfonlayer}{edgelayer}
		\draw [style=diredge] (2) to (3.center);
		\draw [style=diredge] (7.center) to (6);
		\draw [style=diredge] (6) to (2);
		\draw [style=diredge, bend right] (2) to (9);
		\draw [style=diredge, bend right] (9) to (6);
		\draw [style=diredge, bend left] (2) to (11);
		\draw [style=diredge] (11) to (10);
		\draw [style=diredge, bend left] (10) to (6);
	\end{pgfonlayer}
\end{tikzpicture}
    \caption{\small The signal flow graph for the Fibonacci sequence. That means, received as input the sequence $1,0,0,\dots$ of signals, it outputs the Fibonacci numbers $1,1,2,3,5,\dots$. Note these are the coefficients of the formal power series expanding the rational function $\frac{1}{1-x-x^2}$. The semantics of SFG is further explained in Subsection~\ref{subsecsfg}.\label{fig:fibonacci}}
\end{figure}
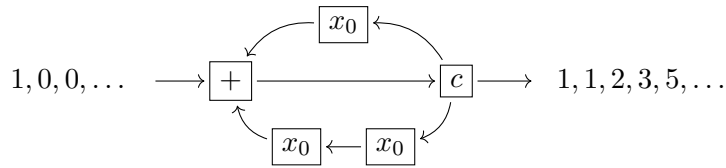

In the last decade, there has been a renewed interest in the theory of signal flow graphs. For readers familiar with category theory, we note that these structures have been studied through the lenses of (co)algebra~\citep{rutten-tutorial,rutten-rational}, logic~\citep{milius-10}, and category theory~\citep{bsz-2017,survey-signal-flow,bsz-2014,baez-erbele}, with the aim of providing a compositional perspective on their behaviour. In particular, these works have established a formal correspondence (in terms of functors) between the syntax of signal flow graphs --- variously represented as a stream calculus, string diagrams, or a graph formalism --- and their semantics --- linear relations between streams.

One aspect that such theory leaves mostly to guesswork and intuition is how to infer the structure of the signal flow graph from the observed behaviour of the system it represents. More precisely, we are interested in the following scenario: given a ``black box" signal flow graph, we may observe the output behaviour only, but not the structure of the graph itself, which we wish to infer in an algorithmic manner.\footnote{Of course, one can only hope to learn the signal flow graph structure up to semantic equivalence. There is a purely equational theory axiomatising this equivalence, as studied in~\citet{bsz-2017,baez-erbele}.} 

Our contribution lies in studying how formal learning techniques~\citep{ANGLUIN198787}, as developed in the context of automata theory, have the potential to improve and systematise this process.  As a starting point, we will develop such a learning algorithm for the restricted case of {\em closed} signal flow graphs (cSFG), where no inputs are provided to the system. In order to adapt existing automata learning techniques, we wish to utilise the fact that cSFGs are semantically equivalent to weighted finite automata with a singleton alphabet.

\emph{Weighted finite automata} (WFA), first introduced in~\citet{schutzenberger1961definition}, can be seen as a generalisation of non-deterministic automata, with each transition having a weight (or a `cost') associated to it. WFAs have found applications in multiple areas ranging from speech recognition~\citep{mohri2005weighted} and image processing~\citep{culik1993image} to control theory~\citep{isidori1985nonlinear} and financial modelling~\citep{baier2009model}. If we restrict the alphabet of a WFA to a singleton one, we get a rather simple class of WFAs. It turns out that such automata are equivalent to cSFGs in the sense that both compute the same subclass of streams (namely, {\em rational streams}).

A learning algorithm for WFA has been introduced by~\citet{bergadano1996learning}, who presented its complexity bound and showed that WFAs over fields could be learned from a teacher. Another such algorithm for WFA was introduced in the PhD thesis of~\citet{gerco}, where it can be viewed as a general weighted adaptation of Angluin’s L$^*$~\citep{ANGLUIN198787} parametric on an arbitrary semiring. In fact, our algorithm is inspired by~\citet{gerco}, though more efficient and with an explicit calculation of the complexity bound.

The complexity bound for our algorithm turns out to be better than that for arbitrary weighted finite automata in the singleton alphabet case. The improvement is largely due to exploiting the linear algebraic structure of rational streams.

\section{Preliminaries}
We recall the notions of rational streams, signal flow graphs and weighted automata. Some of the results stated have references to~\citet{rutten-rational} as the paper nicely condenses them in one place; however, most of them have already been established in the literature prior to it. The proof of rational series having finite representations appears in~\citet{berstel1988rational}. It is further elaborated on with respect to rational streams in~\citet{rutten2001elements}, together with their property of having a finite number of linearly independent derivatives.

\subsection{Rational Streams}

A \emph{stream} over a field $k$ is an infinite sequence of elements of $k$. Define the set of $k$-valued streams as the function set $k^\omega$. Given a stream $\sigma \in k^\omega$, we denote its values by $\sigma = (\sigma_0, \sigma_1, \sigma_2,\dots)$, and call $\sigma_0$ the {\em initial value} of $\sigma$. The \emph{stream derivative} of $\sigma$ is defined as $\sigma' \coloneqq (\sigma_1, \sigma_2, \sigma_3,\dots)$.

We refer to a stream $\sigma$ as \emph{polynomial} if it has a finite number of non-zero elements, i.e.~if it is of the form $\sigma = (c_0, c_1, c_2,\dots, c_k, 0, 0, 0,..)$. A stream is \emph{rational} if it can be represented as a ratio of two polynomial streams, i.e.~a so-called convolution product of a polynomial stream with an inverse of another polynomial stream (see~\citet{rutten-tutorial} for details). As an example, take two polynomial streams $p = (1, 0, 0,..)$ and $q = (1, 0, -1, 0, 0,..)$; then $\sigma = p/q = (1, 0, 1, 0, 1, 0,..)$ is rational. Rational streams are ``periodic", in the sense that they only have a finite number of linearly independent stream derivatives:
\begin{lemma}[\citet{rutten-rational}]
\label{Linear independence}
For a rational stream $\sigma$, there exists an $n \ge 1$ such that $\{\sigma^{(0)}, \sigma^{(1)}, \sigma^{(2)},\dots, \sigma^{(n-1)}\}$ are linearly independent in the vector space $k^n$, and $\sigma^{(n)} = \sum_{i=0}^{n-1} c_i \times \sigma^{(i)}$, for some unique $c_0,\dots,c_{n-1} \in k$.
\end{lemma}

Rational streams have long been used as a means to characterise finite circuits in the field of signal processing~\citep[p.~694]{lathi1998signal}, and can be defined via linear recurrent sequences or difference equations~\citep{rutten2001elements}.

\subsection{Signal Flow Graphs} \label{subsecsfg}

Given a field $k$, we think of its elements as basic units of signals. A {\em signal flow} can then be modelled as an infinite stream of elements from $k$. The transducers that can add, copy, multiply and delay signals are represented as {\em signal flow graphs}.

Formally, a signal flow graph (SFG) is a finite directed graph whose interior vertices are labelled with one of the labels $+$, $c$, $a$ or $x_a$, where $a\in k$. Moreover, a $+$-labelled vertex has at least two incoming edges and one outgoing edge, a $c$-labelled vertex has exactly one incoming edge and at least two outgoing edges, while $a$ and $x_a$-labelled vertices have exactly one incoming and one outgoing edge. Finally, we require that every feedback loop (i.e.~a cycle in the graph) passes through at least one $x_a$-labelled vertex. We draw the four basic \emph{generators} (i.e.~vertex types) of SFGs below; note that the input/output labels are not part of the data of a SFG, but define how it computes a stream function, as discussed below.
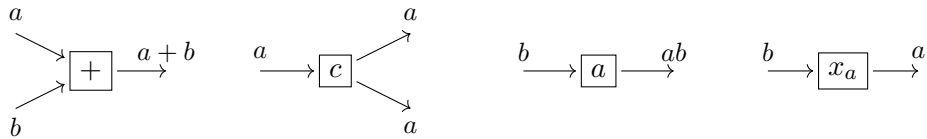
\ctikzfig{sfg-generators}
We refer to the above generators as an {\em adder}, a {\em copier}, a {\em multiplier} and a {\em register}. The vertices with no incoming edges are referred to as {\em inputs} and the vertices with no outgoing edges as {\em outputs}. SFGs are hence built by composing such vertices~\citep{survey-signal-flow}.  Writing $i,o\in\mathbb N$ for the number of inputs and outputs of a signal flow graph, the graph implements a function $\left(k^{\omega}\right)^i\rightarrow \left(k^{\omega}\right)^o$ as follows: at step $n$, the input vertices are assigned the $n$th element of the corresponding stream, the labelled vertices pass on the incoming values to the outgoing values as indicated in the above picture, with the additional requirement that the register value is updated to $x_b$.

Crucially, two SFGs can be composed by plugging (some of) the outputs of one into (some of) the inputs of the other. We have already given an example of a composite SFG computing the Fibonacci sequence (Figure~\ref{fig:fibonacci}).

The class of stream functions computed by SFGs is characterised by multiplication by rational streams~\citep{rutten-tutorial,bsz-2017}. The simplest case of no inputs and one output can be used to characterise rational streams.
\begin{definition}[cSFG]
\label{cSFGs}
A \emph{closed signal flow graph (cSFG)} is a SFG with no inputs and exactly one output.
\end{definition}
We give an example of a cSFG below, which computes the stream $\sigma_n = a+n+1$:
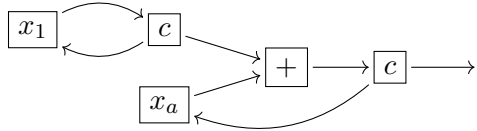
\ctikzfig{csfg-example}
The following result shows that cSFGs precisely capture the rational streams.
\begin{theorem}[\citet{rutten-rational}]\label{Stream to SFG}
A stream $\sigma\in k^{\omega}$ is rational if and only is there exists a cSFG that implements it.
\end{theorem}

\subsection{Weighted Automata}
\emph{Weighted finite automata} (WFA) are a generalisation of non-deterministic automata, where transitions have a weight. A WFA over a field $k$ with an {\em alphabet} $A$ is a tuple $(Q, \langle o,t \rangle)$, where $Q$ is a finite set of {\em states}, $o: Q \rightarrow k$ is an {\em output function} (i.e.~the final weights) and $t: Q \rightarrow (k^Q)^A$ a {\em transition function}~\citep{bonchi2012coalgebraic}.

\begin{remark}
WFA are usually defined as a tuple $(Q, i, o, t)$, where $Q,o,t$ are as above and $i: Q \rightarrow k$ is the initial weights function. We omit the latter from our definition, as the notion of a stream represented by a state is equivalent to assigning input $1$ to it, and $0$ to all other states.
\end{remark}

In the singleton alphabet case, it is equivalent to work with {\em weighted stream automata (WSA)}. The data of a WSA is that of a WFA, except that the transition function has the type $t: Q \times Q \rightarrow k$~\citep{rutten-rational}. An example of a WSA is as below, where we do not draw the transitions labelled with $0$:
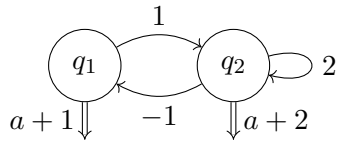
\[
\begin{tikzpicture}
    \node[state] (q_1) {$q_1$};
    \node[state] (q_2)  [right=of q_1] {$q_2$};
    \node (output1) [below =0.5 cm of q_1] {};
    \draw (q_1) edge[-implies, double, double distance=0.5mm] node[left] {$a+1$} (output1);
    \node (output2) [below =0.5 cm of q_2] {};
    \draw (q_2) edge[-implies, double, double distance=0.5mm] node[right] {$a+2$} (output2);
    \path[->] (q_1) edge [bend left]  node [above] {$1$} (q_2);
    \path[->] (q_2) edge [bend left] node [below] {$-1$} (q_1);
    \path[->] (q_2) edge [loop right] node [right] {$2$} (q_2);
\end{tikzpicture}
\]
The weight of a path is obtained by multiplying the weights of the individual transitions in the path and the output weight of the last state. We say that a state $q$ of a WSA {\em represents} a stream $S(q) = (s_0, s_1, s_2,\dots)$, where $s_n$ is the sum of all weights of paths of length $n$ starting at $q$. The state $q_1$ of the above WSA represents the stream $\sigma_n = a+n+1$, which we recognise as the same stream as in our example of a cSFG. In fact, we have the following result:
\begin{theorem}[\citet{rutten-rational}]
\label{Linear System to WSA}
A stream $\sigma \in k^\omega$ is rational if and only if there is a $k$-valued WSA $(Q, \langle o,t \rangle)$ and a state $q\in Q$ representing $\sigma$.
\end{theorem}
To construct a WSA from a rational stream $\sigma$, we first find the number of linearly independent derivatives of $\sigma$, call it $n$; then assign $\sigma_i$ to the output of each state $q_i$ (for $i = 0,\dots,n-1$) and use the coefficients $c_0,\dots, c_{n-1}$ (Lemma~\ref{Linear independence}) to define the transition function. This induces a specific shape of a WSA described in the next section.

\section{The Learning Algorithm}\label{sec:algorithm}
We begin by fixing the definitions needed by the algorithm. We largely follow the notation and terminology from~\citet{gerco}. Throughout this section, let $k$ be a fixed field and $\sigma\in k^{\omega}$ a stream. We think of $\sigma$ as the (unknown) stream computed by the cSFG we want to learn. By an \emph{observation table} we mean a pair $(S, E)$ of finite subsets of $\omega$.

\begin{definition}[Row function, last row function]
Given an observation table $(S, E)$, let $w$ be the largest element of $S$. Define the \emph{row function} $\mathtt{row}: S \rightarrow k^E$ by $\mathtt{row}(v)(e) = \sigma_{v+e}$ for $v \in S$, $e \in E$, and the \emph{last row function} $\mathtt{srow}: E \rightarrow k$ by $\mathtt{srow}(e) = \sigma_{w+e+1}$ for $e \in E$.
\end{definition}

\begin{definition}[Closedness]
\label{closedness}
We say that an observation table $(S, E)$ is \emph{closed} if there exist constants $c_s \in k$ such that $\mathtt{srow} = \sum_{s \in S} c_s \cdot \mathtt{row}(s)$.
\end{definition}

We next introduce the notion of the \emph{coefficient function} that outputs the solution to a system of linear equations if it exists, or returns $\bot$ otherwise. Note that it can be computed efficiently by Lemma~\ref{Linear independence}.

\begin{definition}[Coefficient function]
Let $(S, E)$ be an observation table. Define the \emph{coefficient function} $\cs \in \{\bot\} \cup k^S$ as the function $\cs: S \rightarrow k$ such that $\mathtt{srow} = \sum_{s \in S} \cs(s) \cdot \mathtt{row}(s)$ if the table is closed, and as the symbol $\bot$ otherwise.
\end{definition}

The algorithm assumes access to an oracle that answers the following types of queries: (1) \emph{Membership query}: given an index $n \in \omega$, an oracle replies with the $n$th element of the corresponding stream; and (2) \emph{Equivalence query}: given a cSFG, an oracle replies with ``yes'' if the cSFG constructed is correct, and with ``no'' otherwise.

We now have all the ingredients to define the learning algorithm (Algorithm~\ref{alg:cap}). We create an observation table, with $S$ and $E$ initially only containing $0$. We also keep a counter $i$, initially assigned $1$, to keep track of the smallest index not in $S$, which will improve the overall complexity of the algorithm. We then repeatedly check whether the table is closed (lines 4-8) and if not, add the index $i$ to $S$ and $E$ (i.e.~expanding the table proportionally) and increment the counter. 

We can now construct the hypothesis cSFG (lines 9-14). Thus suppose that the observation table is closed, so that we have a coefficient function $\cs: S \rightarrow k$. Note that the set $S=\{0,\dots,i-1\}$ has $i$ elements. We define the cSFG as follows:
\begin{itemize}
\item it has $i$ registers $r_0,\dots,r_{i-1}$, whose labels are the solutions to the linear equations in Theorem~\ref{thm:correctness},
\item it has $i$ copiers $q_0,\dots,q_{i-1}$,
\item it has $i$ adders $a_0,\dots,a_{i-2}$ and $a_{final}$,
\item it has $2i$ multipliers $m_0,\dots,m_{i-1}$ and $n_0,\dots,n_{i-1}$ with labels $l(m_j)\coloneqq\row(j)(0)=\sigma_{j}$ and $l(n_j)\coloneqq\cs(j)$,
\item for each $j\in\{0,\dots,i-1\}$, there are edges $(r_j,q_j)$, $(q_j,m_j)$ and $(m_j,a_{final})$,
\item for each $j\in\{0,\dots,i-2\}$, there are edges $(q_j,a_j)$ and $(a_j,r_{j+1})$,
\item there are edges $(q_{i-1},n_0)$ and $(n_0,r_0)$, and for each $j\in\{1,\dots,i-1\}$, edges $(q_{i-1},n_j)$ and $(n_j,a_j)$,
\item there is an edge $(a_{final},o)$, where $o$ is the unique output.
\end{itemize}
Below, we give a cSFG graph obtained from the above construction for $i=3$. We denote the label of the vertex inside a box, while its name is put next to it.
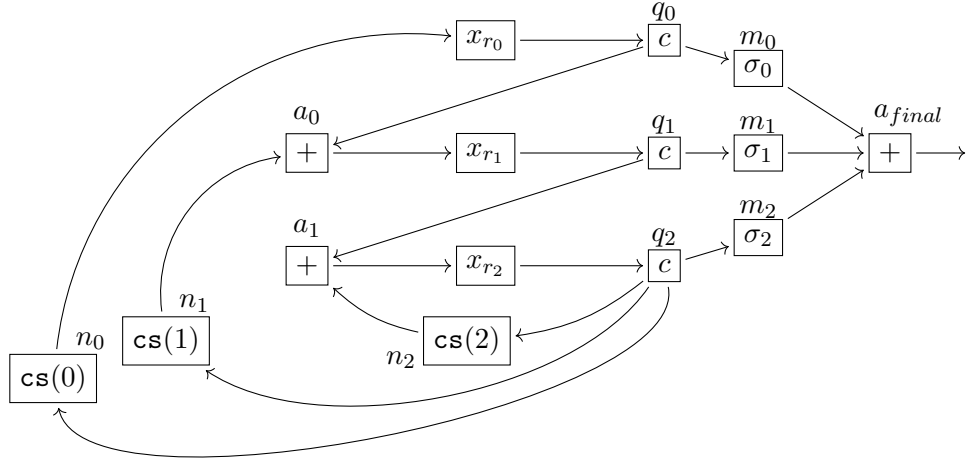
\ctikzfig{csfg-from-rational} \label{csfg3reg}

\begin{remark}
We note that, when constructing a hypothesis cSFG, we essentially get a WSA equivalent to it ``for free'': the values from the first column are assigned to the output of each state $s \in S$; each transition from state $s$ to $s+1$ is assigned value $1$, and the coefficients returned by the coefficient function are assigned to transitions from the state $i-1$ to each $s \in S$. We represent the resulting WSA below:
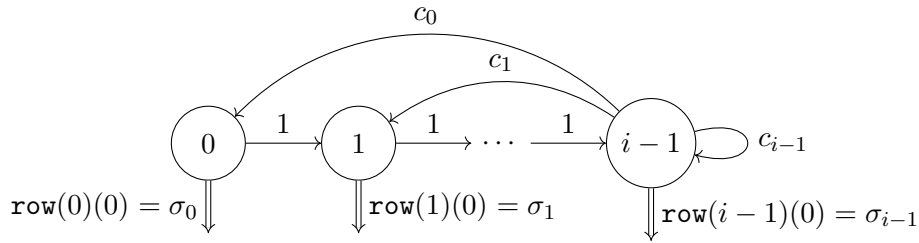
\[
\begin{tikzpicture}
    \node[state] (q_0) {$0$};
    \node[state] (q_1) [right=of q_0] {$1$};
    \node (dots) [right=of q_1] {$\cdots$};
    \node[state] (q_n) [right=of dots] {$i-1$};
    \node (output0) [below =0.7 cm of q_0] {};
    \draw (q_0) edge[-implies, double, double distance=0.5mm] node[left] {$\mathtt{row}(0)(0)=\sigma_0$} (output0);
    \node (output1) [below =0.7 cm of q_1] {};
    \draw (q_1) edge[-implies, double, double distance=0.5mm] node[right] {$\mathtt{row}(1)(0)=\sigma_1$} (output1);
    \node (outputn) [below =0.7 cm of q_n] {};
    \draw (q_n) edge[-implies, double, double distance=0.5mm] node[right] {$\mathtt{row}(i-1)(0)=\sigma_{i-1}$} (outputn);
    \path[->] (q_0) edge node[above] {$1$} (q_1);
    \path[->] (q_1) edge node[above] {$1$} (dots);
    \path[->] (dots) edge node[above] {$1$} (q_n);
    \path[->] (q_n) edge [bend right=35] node [above] {$c_1$} (q_1);
    \path[->] (q_n) edge [bend right=45] node [above] {$c_0$} (q_0);
    \path[->] (q_n) edge [loop right] node [right] {$c_{i-1}$} (q_n);
\end{tikzpicture}
\]
\end{remark}
Finally, we give the hypothesis cSFG to the oracle as an equivalence query, which results in oracle either rejecting it, which indicates that the corresponding stream has more linearly independent derivatives than we have guessed, thence $i$ is added to $S$ and $E$ and the counter is incremented and the outer while loop is executed again; or accepting it and hence we have learnt the cSFG.

\begin{algorithm2e}[ht]
\caption{Abstract learning algorithm for a cSFG over $k$ \label{alg:cap}}
$S,E \gets \{0\}$\\
$i \gets 1$\\
\While{true}{
    \While{$\cs$ = $\bot$}{
        $S \gets S \cup \{i\}$\\
        $E \gets E \cup \{i\}$\\
        $i \gets i + 1$
    }
    \For{$s \in S$}{
        $o(s) \gets \mathtt{row}(s)(0)$\\
        $t(i-1)(s) \gets \cs(s)$\\
        $t(s)(s+1) \gets 1$
    }
    $\text{Construct } cSFG$\\
    \eIf{$\text{EQ}(cSFG) = \text{no}$}{
        $S \gets S \cup \{i\}$\\
        $E \gets E \cup \{i\}$\\
        $i \gets i + 1$
    }{
        \Return{cSFG}
    }
}
\end{algorithm2e}

\subsection{Correctness} \label{subsec:corr}
We proceed to give a proof of the correctness of our algorithm. This amounts to showing that the constructed WFA and cSFG indeed compute the stream we expect them to.

Let us fix the following definitions. A state $q_R$ of a WFA $(Q,\langle o,t \rangle)$ is \emph{reachable} within $R$ steps from state $q_0$ at cost $C_R = \prod_{i=0}^{R} c_i$ if, starting from state $q_0$, there exist $q_1, q_2,\dots,q_{R-1}$ such that $t(q_i)(q_{i+1}) = c_i$ for $i = 0,\dots, R-1$, and $c_{R} = o(q_R)$. Then the stream \emph{computed} by the state $q_0$ is defined as $(s_0, s_1, s_2,\dots)$ such that the element $s_j$ is obtained by adding up the costs of all paths of length $j$ from the state $q_0$ to any reachable state.

\begin{theorem} \label{thm:correctness}
Let $\sigma$ be the hidden stream fixed at the beginning of Section~\ref{sec:algorithm}. The output cSFG returned by Algorithm~\ref{alg:cap} (and hence the corresponding WFA) compute $\sigma$.
\end{theorem}

\begin{proof}
We first prove the statement for the WFA. We know the values $\sigma_0,\dots,\sigma_{n-1}$ from the membership queries. 
Hence, for the first $n$ elements of the stream computed by the state $q_0$ we have:
\begin{equation}
\begin{cases}
      s_0 = C_0 = o(q_0) = \sigma_0 \\
      s_1 = C_1 = t(q_0)(q_1) \times o(q_1) = 1 \times \sigma^{(1)}_0 = \sigma_1 \\
      \vdots \\
      s_{n-1} = C_{n-1} = t(q_0)(q_1) \times \dots \times t(q_{n-2})(q_{n-1}) \times o(q_{n-1}) = 1 \times \dots \times 1 \times \sigma^{(n-1)}_0 = \sigma_{n-1}
\end{cases} \tag*{}
\end{equation}
as expected. Until now, there existed a unique path of length $j$ from $q_0$ to $q_j$, $0 \le j \le n-1$. Let us then show how an element of the stream, namely $s_n$, is calculated when this isn't the case anymore, i.e.~for $j \ge n$:
$$s_n = \sum C_n = \sum_{i=0}^{n-1} t(q_0)(q_1) \times \dots \times t(q_{n-2})(q_{n-1}) \times t(q_{n-1})(q_i) \times o(q_i) = \sum_{i=0}^{n-1} c_i \times \sigma^{(i)}_0$$
which is precisely the definition of our rational stream $\sigma$. This concludes the proof for WFA.

Now, recall that our cSFG has $n$ registers. We can think of each register $r^i$, where $i = 0, \dots, n-1$, as having a certain stream $\upsilon^i$ feeding into it, and a certain stream $\tau^i$ coming out of it. We then have $n$ linear equations in the form $\tau^i = r^i + X \times \upsilon^i$. We obtain the following set of equations: 
\begin{equation}
\begin{cases}
      \tau^0 = r^0 + X \times (c_0 \times \tau^{n-1}) \\
      \tau^1 = r^1 + X \times (\tau^0 + c_1 \times \tau^{n-1}) \\
      \vdots \\
      \tau^{n-1} = r^{n-1} + X \times (\tau^{n-2} + c_{n-1} \times \tau^{n-1})
\end{cases} \tag*{}
\end{equation} 
Notice that we can rewrite each $\tau^i$ in terms of $\tau^{n-1}$ and substitute $\tau^{n-2}$ in the last equation. We then get $$\tau^{n-1} = (r^{n-1}, r^{n-2} + c_{n-1} \times \tau^{n-1}_0, r^{n-3} + c_{n-2} \times \tau^{n-1}_0 + c_{n-1} \times \tau^{n-1}_1, \dots)$$
We can see that the definition of the stream is recursive and depends on the values of $r^i$. We can calculate each of $\tau^i$ in the exact same manner. Recall that the constructed cSFG, based on its structure in the algorithm, computes $\sigma$ if $\sigma_0 \times \tau^0 + \sigma_1 \times \tau^1 + \dots + \sigma_{n-1} \times \tau^{n-1} = \sigma$. Hence, by adding up the corresponding elements $\tau^i_j$ of $n$ streams and equating the result to the known $\sigma_j$, $j = 0, \dots, n-1$, we obtain $n$ equations in $n$ unknowns $r^0,\dots,r^{n-1}$. This way, we have $\sigma_0 \times r^0 + \dots + \sigma_{n-1} \times r^{n-1} = \sigma_0$, and so on. Having obtained the solution, we now know the initial register values of our cSFG, and hence, by definition, the resulting cSFG computes the expected stream $\sigma$.

\end{proof}

\section{Complexity}
We first note that the algorithm always terminates. Computing the coefficient function consists of checking for linear independence of stream derivatives, of which there are finitely many (Lemma~\ref{Linear independence}). Hence, linear dependence is guaranteed to be reached within a finite number of the main loop executions.

Let $n$ be the number of linearly independent derivatives of the stream that the cSFG computes (which is also the number of registers in the cSFG we will learn). Then we have that (1) $|S| \le n$, (2) $|E| \le n$, and (3) the main loop in the algorithm is repeated at most $n$ times. All of these hold by Lemma~\ref{Linear independence}.

Now, let us evaluate complexity bounds of parts of the algorithm inside the main loop. We begin by determining the complexity of the coefficient function (line 4). Since it checks whether the row $i$ is a linear combination of rows corresponding to elements of $S$, each call amounts to solving a system of at most $n$ linear equations in at most $n$ unknowns, as $S$ and $E$ expand evenly. This can be done using Gauss's method with complexity $O(n^3)$~\citep[p.~12]{boyd2018introduction,Farebrother}. Since we make at most $n$ calls to the coefficient function throughout the algorithm (as there are $n$ linearly independent derivatives), we need to solve at most $n$ such systems of equations. So the overall complexity of closing the table is $O(n^4)$. As a side note, the counter on line 7 of the algorithm is introduced precisely to enable calling the coefficient function as infrequently as possible, as it is a task of, as we will deduce later, highest computational complexity.

Next, we consider the complexity of filling the table. In the worst case, the table is of size $|S||E| = n^2$. However, since repeated table entries do not need to be re-queried and each membership query for filling one entry of the table corresponds to a stream index which is at most $n+n = 2n$, the number of membership queries and hence the complexity of filling the table is $O(n)$. As the main loop in the algorithm is executed at most $n$ times, the number of equivalence queries is $O(n)$.

Finally, let us evaluate the complexity of constructing a cSFG (line 14). The multiplier values come from the first $n$ elements of the stream and from the coefficients $c_0,\dots,c_{n-1}$ from Lemma~\ref{Linear independence}; all we have to do now is find the initial register values. There are $n$ registers, so we have to solve $n$ linear equations in $n$ unknowns (this appears in the proof of Theorem~\ref{thm:correctness}). It can be done via Gauss's method of complexity $O(n^3)$.

We conclude that the task of highest complexity is closing the table, $O(n^4)$, which is hence also the overall complexity of our algorithm; i.e. it runs in polynomial time.

\begin{remark} \label{compimrp}
We note that the bound can be marginally improved by a slight modification of the algorithm. Instead of expanding $S$ and $E$ linearly, we could do so exponentially, i.e.~by doubling their sizes whenever the table is not closed and after each unsuccessful equivalence query. This way, we would make at most $O(log \text{ }n)$ equivalence queries and would only make $O(log \text{ } n)$ calls to the coefficient function. The overall complexity of the algorithm would then be $O(n^3log \text{ }n)$.
\end{remark}

Our algorithm has a strictly lower complexity bound than those of learning algorithms for WFAs restricted to the singleton alphabet case. Denote by $m$ the size of the longest counterexample provided by oracles in such algorithms. In~\cite{bergadano1996learning}, the bound is $O(n^5m^2)$ and in~\cite{beimel2000learning} the bound is $O(n^{3.376}+mn^3)$. In the case of the learning algorithms for \emph{weighted tree automata} (WTA) -- a generalisation of WFA, the algorithm by~\citet{maletti2007learning} for deterministic WTA simplifies in our case to having the complexity $O(n^3(n+m))$. We note that the lack of the need for counterexamples in our algorithm allows for the improvement in complexity, as counterexamples are provided non-deterministically and can be arbitrarily large. This, as well as the procedure outlined in Remark~\ref{compimrp}, is the case precisely due to the rational streams' property of having a finite number of linearly independent derivatives.

We also calculated the computational complexity of the algorithm in~\cite{gerco}. The algorithm is for general WFA and relies on the existence of the so-called closedness strategy that checks whether an observation table is closed; it only checks for closedness and not for consistency, as the counterexamples are handled in a way that the constructed observation table is always consistent. Said closedness strategy is a rather generic function, and in order to be able to compare the algorithms, we assume that in the case of a singleton alphabet it would mirror the behaviour of our coefficient function; it is, just like in our algorithm, the task of highest complexity. Hence, as $|E|$ is bounded by $m+1$ (it only expands from suffixes of counterexamples) and $|S|$ is bounded by $n$, we have that the complexity of solving at most $m+1$ equations in at most $n$ variables is $O(n(m+1)\text{ min}(n, m+1)) = O(nm\text{ min}(n,m))$. Since $\cs$ is called when assigning values to the transition function for each state whenever a hypothesis is constructed, and $|S| \le n$ and the outer loop is repeated at most $n$ times, the overall complexity of the algorithm is $O(n^3m\min(n,m))$. However, if we assume that $\cs$ is never recomputed for rows in $S$ during the hypothesis construction even after adding counterexamples to $E$, the complexity bound is $O(n^2m\min(n,m))$. In both cases, our algorithm has a better bound -- as noted earlier, $m$ can be arbitrarily large. We thus conclude that our algorithm genuinely uses the structure of cSFGs, rather than just restricting a generic algorithm for WFAs.

\section{Concluding Remarks}
The main contribution of this work is establishing a learning algorithm for closed signal flow graphs and calculating its computational complexity, considering both the time complexity and the number of queries made. To the best of our knowledge, this work introduces the first algorithm of its kind that allows for learning closed signal flow graphs. Our results show that the proposed algorithm has a better complexity than that for arbitrary weighted automata in~\citet{gerco}, even if the latter one is restricted to WFA over a singleton alphabet. It also has a better computational complexity than learning algorithms for WFA introduced in~\citet{bergadano1996learning} and in~\cite{beimel2000learning}.

An obvious direction for future work is to find a learning algorithm for all SFGs, but that would require proving their equivalence to some family of automata first. In their work,~\citet{basold2014co} show the equivalence between open SFGs and Mealy Machines; the authors of this paper would also wish to establish such equivalence between SFGs and (some subset of) WFAs. An interesting related question is whether learning respects the compositional structure of SFGs: is learning a composite SFG equivalent to learning its parts and composing back?

Automata learning has extensive applications in areas such as network protocol analysis~\citep{comparetti2009prospex}, software verification~\citep{peled1999black} and natural language processing~\citep{knight2009applications}. Learning signal flow graphs opens the door to exploring similar techniques in control theory as well as in the design and analysis of related families of systems, such as digital circuits~\citep{ghica2024fully}.

\section*{Acknowledgements}
The authors would like to thank Wojciech Różowski for reading and providing helpful feedback on this paper. The authors would also like to thank the anonymous reviewers of the Learning and Automata workshop (LearnAut 2024) for their detailed comments and suggestions, especially the ones regarding the reduction in complexity of the algorithm.

\bibliography{bibliography}

\end{document}